\documentclass{amsart}

\usepackage[
	pdftitle={Quantum curve},
	pdfauthor={Paul Norbury},
	ocgcolorlinks,
	linkcolor=linkblue,
	citecolor=linkred,
	urlcolor=linkred]
{hyperref}

\usepackage{amsfonts}
\usepackage{amssymb}
\usepackage{bbm}
\usepackage{bm}
\usepackage{booktabs}
\usepackage[hang,flushmargin]{footmisc} 
\usepackage[left=25mm,right=25mm,top=25mm,bottom=25mm]{geometry}
\usepackage{graphicx}
\usepackage{mathpazo}
\usepackage{microtype}
\usepackage{multicol}
\usepackage{tikz}
     \usetikzlibrary{arrows,shapes}
\usepackage{xcolor}
     \definecolor{linkred}{rgb}{0.6,0,0}
     \definecolor{linkblue}{rgb}{0,0,0.6}

\theoremstyle{plain}
     \newtheorem{theorem}{Theorem}
     \newtheorem{lemma}{Lemma}
     \newtheorem{proposition}{Proposition}[section]
     \newtheorem{conjecture}{Conjecture}

\theoremstyle{definition}
     \newtheorem{example}[proposition]{Example}
     \newtheorem{definition}[proposition]{Definition}
     \newtheorem{remark}[proposition]{Remark}

\newcommand{\N}{\overline{N}}
\newcommand{\res}{\mathop{\mathrm{Res}}}
\newcommand{\bc}{\mathbb{C}}
\newcommand{\bp}{\mathbb{P}}
\newcommand{\bq}{\mathbb{Q}}
\newcommand{\br}{\mathbb{R}}
\newcommand{\bz}{\mathbb{Z}}
\newcommand{\cb}{\mathcal{B}}
\newcommand{\cd}{\mathcal{D}}
\newcommand{\cf}{\mathcal{F}}

\newcommand{\cl}{\mathcal{L}}

\newcommand{\modm}{\mathcal{M}}

\newcommand {\dd}{\mathrm{d}}
\newcommand {\h}{\hbar}
\newcommand {\x}{\widehat{x}}
\newcommand {\uh}{\widehat{u}}
\newcommand {\vh}{\widehat{v}}
\newcommand {\Pw}{\widehat{P}}
\newcommand {\xx}{{\bf{x}}}
\newcommand {\pp}{\bm{p}}

\newcommand {\y}{\widehat{y}}

\newcommand {\opsi}{\overline{\psi}}
\newcommand{\mmu}{\boldsymbol{\mu}}

\newcommand{\rep}{\mathcal{R}}
\newcommand{\tr}{{\rm tr}\hspace{.2mm}}
\newcommand{\Mbar}{{\overline{\mathcal{M}}}}

\DeclareRobustCommand{\stirling}{\genfrac{[}{]}{0pt}{}}

\renewcommand{\arraystretch}{1.25}
\numberwithin{equation}{section}

\begin{document}

\title{Quantum curves and topological recursion}
\author{Paul Norbury}
\address{Department of Mathematics and Statistics, University of Melbourne, Australia 3010}
\email{\href{mailto:pnorbury@ms.unimelb.edu.au}{pnorbury@ms.unimelb.edu.au}}
\thanks{}
\date{\today}

\begin{abstract}
This is a survey article describing the relationship between quantum curves and topological recursion.  A quantum curve is a Schr\"odinger operator-like noncommutative analogue of a plane curve which encodes (quantum) enumerative invariants in a new and interesting way.  The Schr\"odinger operator annihilates a wave function which can be constructed using the WKB method, and conjecturally constructed in a rather different way via topological recursion.
\end{abstract}

\maketitle

\setlength{\parskip}{0pt}
\vspace{-.8cm}
\tableofcontents
\setlength{\parskip}{6pt}

\vspace{-.8cm}
\section{Introduction}  \label{sec:intro}

A {\em quantum curve} of a plane curve $C\hspace{-1mm}=\hspace{-1mm}\{\hspace{-.5mm}(x,y)\hspace{-.5mm}\in\hspace{-.5mm}\bc^2\hspace{-.8mm}\mid\hspace{-.8mm} P(x,y)\hspace{-.8mm}=\hspace{-.8mm}0\}$ is a Schr\"odinger-type linear differential equation 
\begin{equation}   \label{waveq}
\widehat{P}(\x, \y) \, \psi(p,\h)=0
\end{equation}
where $p\in C$, $\h$ is a formal parameter, and $\widehat{P}(\x, \y)$ is a differential operator-valued non-commutative quantisation of the plane curve with $\x=x\cdot$ and $\y=\h\frac{d}{dx}$, so in particular
\begin{equation} \label{commxy}
[\x,\y]=-\h.
\end{equation}
For example, a quantisation of $P(x,y)=y^2-x$ is the Schr\"odinger operator $\widehat{P}(\x, \y)=\y^2-\x=\h^2\frac{d^2}{dx^2}-x$.  The quantum curve can also be a {\em difference equation} if instead we consider $P(e^x,e^y)=0$---see Section~\ref{sec:use}.  

The equation  \eqref{waveq} is understood via the WKB method. In other words we require $\psi(p,\h)$ to be of the form 
\begin{equation}   \label{wavef}
\psi(p,\h)=\exp(\h^{-1}S_0(p)+S_1(p)+\h S_2(p)+\h^2 S_3(p)+...)
\end{equation}
and the $S_k(p)$ are calculated recursively via \eqref{waveq}.  A simple consequence of \eqref{waveq} is that $S_k(p)$ are meromorphic functions on $C$, where $S_0(p)=\int^p ydx$ may be multi-valued.  The operator $\frac{d}{dx}$ acts on meromorphic functions via composition of the exterior derivative followed by division by $dx$, %---which sends meromorphic differentials on $C$ to meromorphic functions on $C$---
and coincides with usual differentiation on an analytic expansion of a meromorphic function with respect to a local coordinate defined by $x$.
%By $S_k(x)$ we mean an expansion in the local coordinate $x$ around a point, and we don't mean that $S_k(p)$ factors through the meromorphic function on the spectral curve $x:C\to\bp^1$.  

{\bf Fundamental question:}  Can we define $S_k(p)$ directly from the plane curve without using the WKB approximation, and in particular produce a natural choice of $\widehat{P}(\x,\y)$?

A conjectural answer to this question is given by
\begin{equation}  \label{exactSk}
S_k(p)=\sum_{2g-1+n=k}\frac{1}{n!}\int^p\int^p...\int^p\omega^g_n(p_1,...,p_n)
\end{equation}
where $\omega^g_n(p_1,...,p_n)$ are multidifferentials for each $g\geq 0$, $n>0$ recursively defined on the curve $C$ via topological recursion which is described in Section~\ref{sec:EO}.  Curves of genus $g$ with $n$ labeled points give a convenient way to encode topological recursion, and often represent  an underlying geometric connection.  From this viewpoint, \eqref{exactSk} claims that $S_k(p)$ is related to all punctured curves of Euler characteristic $1-k$.  This conjecture is addressed by Gukov and Su\l kowski in \cite{GSuApo} together with the related issue of constructing $\widehat{P}(\x,\y)$ algorithmically from the wave function.

The path from the quantum curve to the plane curve is well-defined.  It is a little deeper than simple substitution $\x\mapsto x$ and $\y\mapsto y$ into $\widehat{P}(\x, \y)$, since we {\em deduce} that the differential  equation \eqref{waveq} is satisfied only on the plane curve $P(x,y)=0$.  This is achieved via the semi-classical limit $\h\to 0$, where the differential operator $\widehat{P}(\x,\y)$ reduces to a multiplication operator that vanishes precisely on the plane curve.  The action of $\h\frac{d}{dx}$ on
$$\psi_0(p,\h)=\exp(\hbar^{-1}\int^p ydx)$$
is multiplication by $y$, so 
$$ \widehat{P}(\x, \y) \, \psi(p,\h)=\Big[P(x,y)+O(\h)\Big]\psi(p,\h)
$$
and in the $\h\to 0$ limit $\y=\h\frac{d}{dx}$ in $\widehat{P}(\x, \y)$ is replaced by its symbol $y$.  Higher order corrections in $\h$ are required since $(\h\frac{d}{dx})^2\mapsto y^2+O(\h)$ under its action on $\psi_0(x,\h)$.

On the other hand, constructing the quantum curve from the plane curve is not canonical. The main issues lie in the construction of the wave function $\psi(p,\h)$ and the ambiguity in ordering the non-commuting operators $\x$ and $\y$ in $\widehat{P}$.  The conjectural formula \eqref{exactSk} is one attempt to remedy this.  Such a wave function is enough to reconstruct the operator $\widehat{P}(\x,\y)$.  Any $\widehat{P}(\x,\y)$ can be expressed as:
\begin{equation} \label{normord}
\widehat{P}(\x,\y)=P(\x,\y)+\h P_1(\x,\y)+\h^2 P_2(\x,\y)+...
\end{equation}
where each $P_k(\x,\y)$ is a normal ordered operator valued polynomial---all $\y$ terms in a monomial are placed to the right, so $x^my^n$ is replaced by $\x^m\y^n$---and has no explicit $\h$ dependence.  Then the $P_k(\x,\y)$ can be reconstructed recursively from the wave function.

The differential operator $\widehat{P}(\x, \y)$ generates a principal ideal in the algebra $\cd$ of differential operators which act on $\bc[x]$.  The quotient $\cd/\langle\widehat{P}\rangle$ of the algebra $\cd$ by the principal ideal $\langle\widehat{P}\rangle=\cd\cdot \widehat{P}$ is a $\cd$-module which gives a way to study $\widehat{P}(\x, \y)$ intrinsically.  See \cite{HolTop} for a detailed description of this.  The wave function $\psi(p, \h)$ can be retrieved via the $\cd$-module homomorphism it defines:
$$\cd/\langle\widehat{P}\rangle\to\bc[[x,x^{-1},\h,\h^{-1}]].$$

%Any plane curve is a Lagrangian curve in the plane, and considered to be classical and is known as a {\em spectral curve}.  The differential $ydx$ has significance... 

\subsection{Model enumerative problem.}
Plane curves and quantum curves naturally arise out of various enumerative problems.  A model problem is the enumeration of moments of a given probability measure.   Given a measure $\rho$ on $\br$ that is well-behaved, say it is bounded with compact support $K\subset \br$, its Stieltjes transform is analytic in $\overline{\bc}-K$,
$$\hat{\rho}(x)=\int_{\br}\frac{\rho(t)dt}{x-t}=\sum_{n\geq 0}\frac{\langle t^n\rangle}{x^{n+1}},\quad \langle t^n\rangle=\int_{\br}t^n\rho(t)dt$$
where the sum is an analytic expansion at $x=\infty$, with coefficients giving the moments of $\rho$.  The function $\hat{\rho}(x)$ extends to a Riemann surface which is a cover of the $x$-plane.  The Riemann surface is equipped with two functions $x$ and $y=\hat{\rho}(x)$ and hence naturally maps to $\bc^2$.  %The plane curve is algebraic for nice enough $\rho$.  
A toy example of this setup is as follows.

\begin{example}  \label{ex1}
Consider the discrete measure 
$$\rho(t)=\sum_{i=1}^N\delta_{\lambda_i}$$
with Stieltjes transform
$$\hat{\rho}(x)=\sum_{i=1}^N\frac{1}{x-\lambda_i}=\tr(x-M)^{-1}$$ 
for $M$ a matrix with eigenvalues $\{\lambda_1,...,\lambda_N\}$ conveniently chosen to have resolvent $\hat{\rho}(x)$.  The Stieltjes transform is a holomorphic function on a plane curve
which we call its {\em spectral curve}
\begin{equation} \label{spec1}
\left\{y-\sum_{i=1}^N\frac{1}{x-\lambda_i}=0\right\}.
\end{equation}

The function
$$\psi(x)=\det(x-M)=\displaystyle\prod_{i=1}^N(x-\lambda_i)=\exp\int^x ydx$$
satisfies the first order differential equation 
\begin{equation} \label{qc1}
\left(\frac{d}{dx}-\sum_{i=1}^N\frac{1}{x-\lambda_i}\right)\psi(x)=0
\end{equation}
which is the {\em quantum curve}, a non-commutative analogue of the spectral curve since $(x,y)$ is replaced in the equation of the curve \eqref{spec1} by operators $(\x=x\cdot,\y=\frac{d}{dx})$ to produce the differential operator in \eqref{qc1}.

If we introduce $\h$ into the wave function via
$\psi(x)=\exp\h^{-1}\int^x ydx$ and $\y=\h\frac{d}{dx}$, then it satisfies the quantum curve equation $(\hat{y}-\sum_{i=1}^N\frac{1}{x-\lambda_i})\psi(x)=0$ exactly but the perturbative parameter $\h$ is not needed here.
\end{example}
\begin{remark}
The example above is a special case of the more general construction of normal ordered first order quantum curves associated to plane curves linear in $y$, i.e. $P(x,y)=p(x)+q(x)y$.  In these cases the first order approximation of the quantum curve gives the entire quantum curve.
\end{remark}
\begin{remark}
The elementary relation of this example to the spectrum of a matrix anticipates nicely the connection with matrix models which are mentioned briefly in Section~\ref{sec:ex}.  Its relation to the spectrum of a matrix means it also has an elementary relation to symmetric polynomials.  The Stieltjes transform and wave function are generating functions for power sum symmetric polynomials $p_k$, respectively elementary symmetric polynomials $\sigma_k$:
$$\hat{\rho}(x)=\sum_{k=0}^{\infty}\frac{p_k(\lambda_1,\dots,\lambda_N)}{x^{k+1}},\quad \psi(x)=\sum_{k=0}^{\infty}\frac{\sigma_k(\lambda_1,\dots,\lambda_N)}{x^k}
$$
where the second sum is of course finite.  The quantum curve \eqref{qc1} gives
$\frac{d}{dx}\psi(x)=\hat{\rho}(x)\psi(x)$ which is exactly Newton's identities relating $p_k$ and $\sigma_k$.
\end{remark} 

\subsection{WKB method}
The quantum curve defines a triangular system in $\frac{d}{dx}S_k(p)$ known as the WKB method.  The function $\frac{d}{dx}S_k(p)$ first appears in the system as 
$$ \h^k \frac{\partial P}{\partial y}(x,y)\frac{d}{dx}S_k(p)+...
$$
so in particular it is a meromorphic function on the curve $P(x,y)=0$ with poles at the zeros of $dx$.
Here we present an example in order to go through the WKB method explicitly.

\begin{example}  \label{ex2}
Consider the measure
$$\rho(t)=\frac{1}{2\pi}\sqrt{4-t^2}\cdot\chi_{[-2,2]}$$
which can arise as the limit $ \lim_{N\to\infty}\frac{1}{N}\sum_{i=1}^N\delta_{\lambda_i}$ of a normalised version of Example~\ref{ex1}. % with spectral curve: $\displaystyle \left\{y-\sum_{i=1}^k\frac{\rho_i}{x-\lambda_i}=0\right\}$ and quantum curve: $\displaystyle \left(\frac{d}{dx}-\sum_{i=1}^k\frac{\rho_i}{x-\lambda_i}\right)\psi(x)=0$. Normalise the previous example $$\rho(t)=\frac{1}{N}\sum_{i=1}^N\delta_{\lambda'_i}=\sum_{i=1}^k\rho_i\delta_{\lambda_i},\quad\sum_{i=1}^k\rho_i=1$$ $\displaystyle \hat{\rho}(x)=\sum_{i=1}^k\frac{\rho_i}{x-\lambda_i}=\frac{1}{N}\tr(x-M)^{-1}$ for $M$ an $N\times N$ matrix with eigenvalues $\{\lambda_1,...,\lambda_k\}$ of multiplicity $\{N\rho_i\}$. $\psi(x)=\displaystyle\prod_{i=1}^k(x-\lambda_i)^{\rho_i}=\exp\int^x \hat{\rho}(x')dx'=\det(x-M)^{1/N}$

Its Stieltjes transform %is analytic in $\overline{\bc}-[-2,2]$, and 
extends to a meromorphic function on a rational curve realised as a double cover of $\overline{\bc}$ (= the $x-$plane) branched over $x=\pm 2$:
$$\hat{\rho}(x)=\frac{1}{2\pi}\int^2_{-2}\frac{\sqrt{4-t^2}}{x-t}dt=\sum_{n\geq 0}\frac{C_n}{x^{2n+1}}=y,\quad x=y+\frac{1}{y}$$
where the odd moments vanish and $C_n=\frac{1}{n+1}\binom{2n}{n}$ is the $n$th Catalan number. The series in $x$ is a local statement---it is an analytic expansion of the global meromorphic function $y=\hat{\rho}(x)$ in the local coordinate $x$ on a branch of $x=\infty$. (It does not factor through $x:C\to\bc$.) Together $x$ and $y$ define a plane curve %which is the  spectral curve
$$C=\{(x,y)\mid P(x,y)=y^2-xy+1=0\}.$$
The wave function
$\psi_0(p)=\exp\int^p ydx=\exp\int^yy(1-\frac{1}{y^2})dy=\frac{1}{y}\exp\frac{1}{2}y^2$
satisfies $\frac{d}{dx}\psi_0(p)=y\cdot\psi_0(p)$ hence on the spectral curve $C$
$$\left(\frac{d^2}{dx^2}-x\frac{d}{dx}+1\right)\psi_0(p)=(y^2-xy+1+\frac{d}{dx}y(p))\cdot\psi_0(p)=\frac{d}{dx}y(p)\cdot\psi_0(p)\neq 0$$
since $\frac{d}{dx}y(p)=\frac{y^2}{1-y^2}\neq 0$.  We see that when the linear operator is not first order its failure to be as simple as Example~\ref{ex1} is due to higher derivatives.

Introduce a perturbative variable $\h$ via $\y=\h\frac{d}{dx}$ and put $\psi_0(p)=\exp\h^{-1}\int^p ydx$ then
$$\left(\h^2\frac{d^2}{dx^2}-x\h\frac{d}{dx}+1\right)\psi_0(x)=\h \frac{y^2}{1-y^2}\psi_0(x)=O(\h)$$
which is zero up to order $\h$.   Following \eqref{wavef} we can remove all higher terms in $\h$ by replacing $\psi_0$ with
$\psi(p,\h)=\exp\Big(\h^{-1}\sum_{k\geq 0} \h^kS_k(p)\Big)$
where $\frac{d}{dx}S_0(p)=y$ and $S_k(p)$ are chosen so that $P(\x,\y)\psi(p,\h)=0$ for $\x=x\cdot$ and $\y=\h\frac{d}{dx}$ up to all orders of $\h$.
Concretely,
\begin{align*}
0=\psi^{-1}\left(\h^2\frac{d^2}{dx^2}-x\h\frac{d}{dx}+1\right)\psi=&\Big[\frac{d}{dx}S_0(p)^2-x\frac{d}{dx}S_0(p)+1\Big]\\
&+\h\Big[\left(\frac{d}{dx}\right)^2S_0(p)+(2\frac{d}{dx}S_0(p)-x)\frac{d}{dx}S_1(p)\Big]\\
&+\h^2\Big[\left(\frac{d}{dx}\right)^2S_1(p)+\frac{d}{dx}S_1(p)^2+(2\frac{d}{dx}S_0(p)-x)\frac{d}{dx}S_2(p)\Big]\\
&+\h^3\Big[\left(\frac{d}{dx}\right)^2S_2(p)+2\frac{d}{dx}S_1(p)\frac{d}{dx}S_2(p)+(2\frac{d}{dx}S_0(p)-x)\frac{d}{dx}S_3(p)\Big]\\
&+O(\h^4).
\end{align*}
The system is triangular so one can recursively solve for $\frac{d}{dx}S_k(p)$.  For example
\begin{equation}  \label{ex2s1}
\frac{d}{dx}S_1(p)=-\frac{\left(\frac{d}{dx}\right)^2S_0(p)}{2\frac{d}{dx}S_0(p)-x}=-\frac{y^3}{(y^2-1)^2}.
\end{equation}
Furthermore, $\frac{d}{dx}S_k(p)$ is rational in $\frac{d}{dx}S_m(p)$ for $m<k$ and $\frac{d}{dx}S_0(p)$ only appears in the denominator $2\frac{d}{dx}S_0(p)-x=\frac{y^2-1}{y}$.  Since $\frac{d}{dx}S_1(p)$ is a meromorphic function on $C$ with poles only at $y=\pm 1$ then $\frac{d}{dx}S_k(p)$ is a meromorphic function on $C$ with poles only at $y=\pm 1$ for $k>0$.  %As explained above, the dependence of $S_k$ on $x$ here refers to an analytic expansion around $x=\infty$.  The function $S_k$ does not factor through $x:C\to\bc$, for example $S_1'(x)$ which really means $\frac{d}{dx}S_1$ is not invariant under $y\mapsto 1/y$.
\end{example}

For general curves $S_1(p)$ is discussed in \cite{GSuApo}.  For rational curves one obtains:
\begin{equation}  \label{exactS01x}
S_1(p)=-\frac{1}{2}\log\frac{dx}{dz}
\end{equation}
where $z$ is a global rational parameter on the curve $C$.  For example, when $x=z+\frac{1}{z}$, $S_1(z)=-\frac{1}{2}\log(1-\frac{1}{z^2})$ and hence $\frac{d}{dx}S_1(z)=\frac{1}{x'(z)}\frac{d}{dz}S_1(z)=\frac{-z}{(z^2-1)^2}$ which agrees with \eqref{ex2s1} for $y=1/z$.

More generally, for any spectral curve $C$ and local parameter $z$ on $C$, \eqref{exactS01x} gives a first approximation to $S_1(p)$ and involves further terms.

%More generally, for any $P(x,y)$ the vanishing of the coefficient of $\h$ gives $$\frac{\partial P}{\partial y}(x,y)\dfrac{d}{dx}S_1(p)+P_1(x,y)+\frac{1}{2}\frac{\partial^2 P}{\partial y^2}(x,y)=0 $$ and $\frac{\partial P}{\partial y}\not\equiv 0$ on $P(x,y)=0$ (since we assume the curve is reduced) so $\dfrac{d}{dx}S_1(p)$ is a meromorphic function on the curve $P(x,y)=0$ with poles at the zeros of $dx$. 

\subsection{Relations between quantum curves and topological recursion}

There are compelling reasons for conjecturing the relation between quantum curves and topological recursion given by \eqref{exactSk}. There is the close relationship between the disk and annulus invariants which are essentially the input data for topological recursion and the quantum curve.   Also, many examples have been verified, which we will discuss in Section~\ref{sec:ex}.  We discuss below two more properties of the $S_k$ shared by topological recursion and the quantum curve.  The first is the invariance of $S_k(p)$ under a class of isomorphisms between plane curves, where $S_k(p)$ is constructed using either the WKB method or topological recursion.  The second is the local behaviour near poles of topological recursion and the quantum curve which exploits the nice fact that the algebra of operators is in some sense commutative near poles.

\subsubsection{Invariance of $S_k$ under isomorphisms}
Consider the following isomorphism between plane curves 
\begin{equation}   \label{isom}
(x,y)\mapsto (x,y+g'(x))
\end{equation}
for any polynomial $g(x)$ where $g'(x)=\frac{d}{dx}g(x)$.  So their defining polynomials $P(x,y)=0$ and $Q(x,y)=0$ are related by $Q(x,y)=P(x,y-g'(x))$.  Now 
$$\widehat{P}(\x,\y)\psi(p,\h)=0\quad\Rightarrow\quad\widehat{Q}(\x,\y)\exp(\h^{-1}g(x))\psi(p,\h)=0
$$
for $P(x(p),y(p))=0$ and $Q(x(p),y(p)+g'(x(p)))=0$
where $\widehat{Q}(\x,\y)$ has to be defined carefully as follows: replace each appearance of $\y$ in $\widehat{P}(\x,\y)$ with the operator $\y-g'(x)$ and in particular do not normal order.  

The isomorphism \eqref{isom} preserves the underlying curve, not its embedding, together with the function $x$ defined on the curve.   The change in wave function for curves related by such an isomorphism only affects the $\h^{-1}$ term in the exponent of $\psi(p,\h)$ and all $S_k(p)$ for $p>0$ are unchanged under the isomorphism.  So we see that $S_k(p)$ for $k>0$ are in some sense intrinsic  to the underlying curve equipped with the functions $x$ and $y$.  

Moreover, the $\omega^g_n$, which are defined in Section~\ref{sec:EO} are unchanged under isomorphisms of type \eqref{isom}, lending weight to the conjecture  \eqref{exactSk}.

\subsubsection{Local factorisation.}  \label{sec:locfac}
A fundamental plane curve, known as the Airy curve, is $y^2-x=0$.  Its quantum curve does indeed satisfy \eqref{exactSk} \cite{ZhoInt}.  It gives a local model for any curve with $dx$ having simple zeros.  We will see below that the operator $\widehat{P}(\x,\y)$ has a factor of $(y-y_0)^2-\lambda(x-x_0)=0$ which annihilates part of the wave function corresponding to the wave function of the Airy curve.  This uses the fact that the algebra of operators is in some sense commutative near a zero of $dx$.

The WKB method shows that for $k>1$ the function $S_k(x)$ has a pole of order $3k-3$ at any zero $a$ of $dx$ and $S_1(x)$ has a logarithmic singularity there.  We define the largest order term of the $\h$ expansion $\h^{-1}S_0(p)+S_1(p)+\h S_2(p)+\h^2 S_3(p)+...$ at $a$ to be a $\h$ expansion with coefficients the largest order terms of each $S_k(p)$.  

Fix a zero $a$ of $dx$ and consider only the largest order terms at $\alpha$ in the exponent of $\psi$.  The action of the operator $\widehat{P}(\x,\y)$ on highest order terms is rather simple since all operators commute!  They behave like differential operators with constant coefficients.  For example, in a local coordinate $z$
$$\h\frac{d}{dx}x\frac{1}{(z-z(a))^m}=\h\frac{1}{(z-z(a))^m}+\h x\frac{d}{dx}{(z-z(a))^m}=\h x\frac{d}{dx}{(z-z(a))^m}+\text{ lower order terms}
$$
since $\frac{d}{dx}{(z-z(a))^{-m}}=O((z-z(a))^{-m-2}$.  This does not contradict $[x,\h\frac{d}{dx}]=-\h$. Multiplication by any function analytic in $x$, such as $\h$, acts next to $\frac{d}{dx}$ like zero on the highest order parts since an analytic function never increases the order of a pole whereas $\frac{d}{dx}$ always increases the order of a pole by 2.

The polynomial factorises as:
$$ 
P(\alpha,y)=\lambda\prod_{k=1}^n(y-\lambda_k(x))
$$
where $\lambda_k(x)$ are locally defined analytic functions.  Exactly two of the $\lambda_k(x)$ coincide when evaluated at a zero $a$ of $dx$---we may assume $\lambda_1(a)=\lambda_2(a)$.  With respect to a local coordinate $s$ near $a$ defined by $x=a+s^2$, we have $y=y(0)+y'(0)s+O(s^2)$.   Equivalently, $\lambda_1(x(s))=y(0)+y'(0)s+O(s^2)$ and $\lambda_2(x(s))=y(0)-y'(0)s+O(s^2)$.  Hence $(y-\lambda_1(x))(y-\lambda_2(x))=(y-y(0))^2-y'(0)^2s^2+O(s^3)$ and by commutativity near $a$ it can be brought forward and must annihilate the highest order part.  The conclusion is that the highest order part of the wave function is given by a rescaled Airy wave function.

The poles of the invariants $\omega^g_n$ occur at the zeros of $dx$ and the highest order part of $\omega^g_n$ near a pole is given by $y'(0)^{2-2g-n}\omega^{g\text{\ Airy}}_n$. Hence the the highest order part of the wave function using \eqref{exactSk} is $\psi^{\text{\ Airy}}(p,\frac{\h}{y'(0)})$ which agrees with the behaviour of the WKB method.

\subsection{Why are quantum curves useful?}  \label{sec:use} 

One application of quantum curves is to {\em predict} topological recursion.  The proofs are often easier for quantum curves than for topological recursion.  Proofs of both are often equivalent to recursions between enumerative invariants, and the former involves coarser, hence simpler, invariants.  A quantum curve assembles enumerative information in an Euler characteristic expansion using a single variable.  Whereas topological recursion produces several variable invariants with genus expansions.
Furthermore, the nonlinear behaviour of topological recursion which arises out of connectedness assumptions can be simpified to linear behaviour since a wave function satisfies a linear differential equation.  

Two examples of this predictive behaviour are as follows.  Quantum curves for a family of enumerative examples, so-called {\em hypermaps}, were proven in \cite{DMaQua} and there it was conjectured that topological recursion applied to the associated plane curves enumerated hypermaps.  The conjecture was known to be correct since it could be verified in low genus numerical calculations.  This conjecture was later proven in \cite{DOPSCom}.  Currently there is an outstanding conjecture regarding spin Hurwitz numbers.  It is known by numerical verification and the quantum curve was proven in \cite{MSSSpe}.

Another application of the quantum curve should be to enable one to drop technical assumptions on topological recursion.  Topological recursion does not apply to any plane curve.  It requires the zeros of $dx$ to be simple.  Since the quantum curve requires no such assumptions one would expect to be able to makes sense of topological recursion without the technical assumptions.  One such construction is given in \cite{BEyThi}.  This would be useful for curves where $x$ is defined via the quotient of a group action on the curve.

Before we describe the next application we will describe variations on the basic setup of plane curves.  This paper is mainly concerned with curves in $\bc^2$.  The $A$-polynomial describes curves 
$$C\subset\bc^*\times\bc^*$$ with coordinates $(x,y)=(e^u,e^v)$.  If we put $\uh=u\cdot$ and $\vh=\h\frac{d}{du}$ so $[\uh,\vh]=-\h$, then in $(x,y)$ coordinates we have $\x=e^{\uh}=x\cdot$ and $\y=e^{\vh}=e^{\h x\frac{d}{dx}}$ so that
$$\x\y=q^{-1}\y\x,\quad q=e^{\h}$$
and $\y f(x)=e^{\h x\frac{d}{dx}}f(x)=f(qx)$.  Topological recursion---defined in Section~\ref{sec:EO}---is sometimes modified depending on whether one works in $(u,v)$ or $(x,y)$ variables.

We also consider the case of $\bc\times\bc^*$ with coordinates $(u,v)=(x,e^y)$, for example the spectral curve of Gromov-Witten invarants of $\bp^1$ treated in Section~\ref{sec:ex}, where the commutator relations are $[\x,\y]=-\h\y$.

One of the most famous applications of the quantum curve is its relation to the volume conjecture in 3-manifold topology.  Conjecturally the behaviour of the quantum curve of the $A$-polynomial generalises the volume conjecture.  Given a knot $K\subset S^3$, its $A$-polynomial $A_K(m,\ell)=0$ is defined via its $SL(2,\bc)$ representation variety $\rep(G)=\{\rho\mid\rho:G\to SL(2,\bc)\}/\hspace{-1mm}\sim$.  The boundary $T^2=\partial(S^3-K)$ induces a restriction map 
$$\rep_G(S^3-K)\to\rep_G(T^2)\cong\bc^*\times\bc^*\ni \left(\begin{array}{cc}m&0\\0&m^{-1}\end{array}\right),\left(\begin{array}{cc}\ell&0\\0&\ell^{-1}\end{array}\right)$$
with image a Lagrangian curve defined by $A_K(m,\ell)=0$.  Here $\log\ell=v=S_0'(u)$ is known as the Neumann-Zagier potential \cite{NZaVol}.

The coloured Jones polynomial \cite{MMuCol} $J_N(K;q)\in\bz[q,q^{-1}]$ has asymptotics as $N\to\infty$ with $N\h$ fixed and $q=e^{-\h}$ given by
\begingroup
\renewcommand*{\arraystretch}{0.8}
$$\log J_N(K;e^{-\h})\hspace{-8mm}\small{\begin{array}[t]{c}\sim\\N\to\infty\\{\h}\to 0\\N{\h}=u+2\pi i\end{array}}\hspace{-8mm}\h^{-1}S_0(u)+\sum_{k=0}^{\infty}\h^kS_{k+1}(u).
$$
\endgroup
\begin{conjecture}
[Kashaev \cite{KasHyp}, Murakami-Murakami \cite{MMuCol}, Gukov \cite{GukThr}]
For the curve defined by $A_K(m,\ell)=0$, put $(m,\ell)=(e^u,e^v)$ and define $S_0(u)$ via \eqref{waveq} and \eqref{wavef}.  Then
$$S_0(u)= \text{\ volume of (incomplete) hyperbolic manifold\ }S^3-K.$$
\end{conjecture}
\begin{conjecture}
[Dikjgraaf-Fuji-Manabe \cite{DFMVol}]
Consider the quantisation $\hat{A}_K$ that comes out of topological recursion applied to $A_K(m,\ell)=0$.  In other words, use the $S_k(u)$ calculated via topological recursion and \eqref{exactSk} to construct the wave function which can be used to produce $\hat{A}_K$.  Then
$$\hat{A}_KJ_N(K;e^{\h})=0.$$
\end{conjecture}

Dimofte \cite{DimQua} has a beautiful approach to the quantum curve of the $A$-polynomial.  He shows that the character variety of an ideal hyperbolic tetrahedron gives rise to the curve $\{y+x^{-1}-1=0\}\subset\bc^*\times\bc^*$ with quantum curve $(\y+\x^{-1}-1)\psi(x,\h)=0$.  As described above, $\y f(x)=f(qx)$ for $q=e^{\h}$ so the quantum curve equation can be written 
$$\psi(qx,\h)=(1-x^{-1})\psi(x,\h)$$ 
which is satisfied by a quantum dilogarithm function---a building block of link invariants \cite{KasLin}.   A hyperbolic three-manifold can be built by gluing hyperbolic tetrahedra and Dimofte studies how such gluing affects the quantum curve in order to build up the quantum curve of the $A$-polynomial. 

{\em Acknowledgements.} The author benefited from conversations with numerous people, and in particular would like to thank Ga\"{e}tan Borot, Norman Do, Petya Dunin-Barkowsky, Bertrand Eynard, Peter Forrester, John Harnad, Motohico Mulase, Sergey Shadrin, and participants of String-Math 2014 and the Banff workshop {\em Quantum Curves and Quantum Knot Invariants}.

\section{Topological recursion}  \label{sec:EO}

Topological recursion as developed by Chekhov, Eynard and Orantin \cite{CEyHer,EOrInv} arose out of loop equations satisfied by matrix models.  It takes as input a \emph{spectral curve} $(C,B,x,y)$ consisting of a compact Riemann surface $C$, a bidifferential $B$ on $C$, and meromorphic functions $x,y:C\to\bc$.   A technical requirement is that the zeros of $\dd x$ are simple and disjoint from the zeros of $\dd y$~\cite{EOrInv}.  In the case when $C$ is rational with global rational parameter $z$, $B= \frac{\dd z_1 \otimes \dd z_2}{(z_1-z_2)^2}$ is the Cauchy kernel.    

The output of topological recursion is a collection of multidifferentials $\omega^g_n(p_1, \ldots, p_n)$ for integers $g\geq 0$ and $n \geq 1$, on $C$ --- in other words, a tensor product of meromorphic 1-forms on the product $C^n$, where $p_i\in C$. When $2g-2+n>0$, $\omega^g_n(p_1, \ldots, p_n)$ is defined recursively in terms of local information around the poles of $\omega^{g'}_{n'}(p_1, \ldots, p_{n'})$ for $2g'+2-n' < 2g-2+n$. %Equivalently, the $\omega^{g'}_{n'}(p_1, \ldots, p_{n'})$ are used as kernels on the Riemann surface. %If $C$ is compact and comes with a {\em Torelli marking}---a choice of symplectic basis $\{a_i, b_i\}_{i=1,\ldots,g}$ of the first homology group $H_1(\bar{C})$ of the compact closure $\bar{C}$ of $C$, which is an empty condition on a rational surface---then one can define canonical choice of bidifferential $B$ known as the Bergmann kernel.

% ND: what does it mean to say that they are used as kernels on the Riemann surface?

The invariants $\omega^g_{n}(p_1,...,p_n)$ are defined recursively from simpler $\omega^{g'}_{n'}$ for $2g'-2+n'<2g-2+n$.  The recursion can be represented pictorially via different ways of decomposing a genus $g$ surface with $n$ labeled boundary components into a pair of pants containing the first boundary component and simpler surfaces. 
\begin{center}
\includegraphics[scale=0.12]{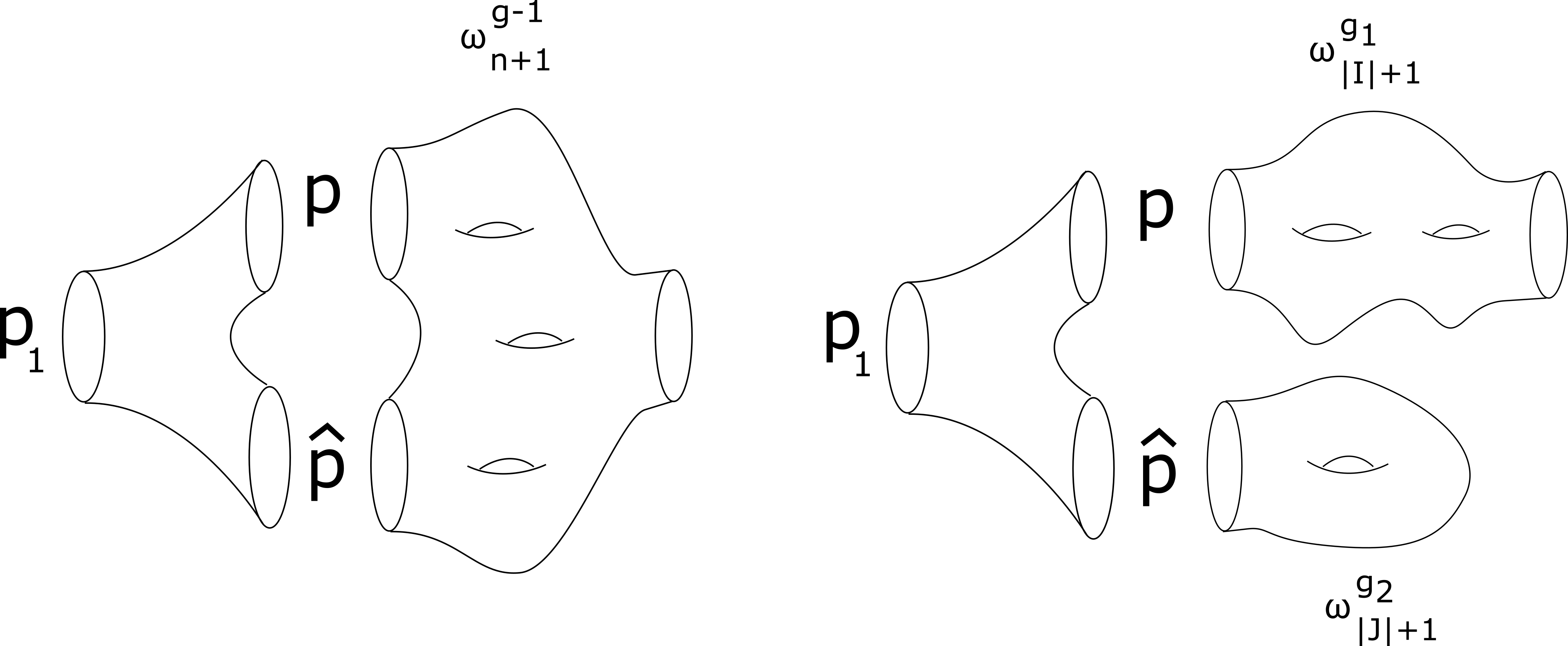}
\end{center}
For $2g-2+n>0$ and $S = \{2, \ldots, n\}$, define
\begin{equation}  \label{EOrec}
\omega^g_{n}(p_1,\pp_{S})=\sum_{\alpha}\res_{p=\alpha}K(p_1,p) \bigg[\omega^{g-1}_{n+1}(p,\hat{p},\pp_{S})+ \mathop{\sum_{g_1+g_2=g}}_{I\sqcup J=S}^\circ \omega^{g_1}_{|I|+1}(p,\pp_I) \, \omega^{g_2}_{|J|+1}(\hat{p},\pp_J) \bigg]
\end{equation}
where the outer summation is over the zeros $\alpha$ of $\dd x$ and the $\circ$ over the inner summation means that we exclude terms that involve $\omega_1^0$.  The point $\hat{p}\in C$ is defined to be the unique point $\hat{p}\neq p$ close to $\alpha$ such that $x(\hat{p})=x(p)$.  It is unique 
since each zero $\alpha$ of $\dd x$ is assumed to be simple, and we need only consider $p\in C$ close to $\alpha$.  We see that the recursive definition of $\omega^g_n(p_1, \ldots, p_n)$ uses only local information around zeros of $\dd x$.  The recursion takes an input the base cases
\[
\omega^0_1=-y(z)\,\dd x(z) \qquad \text{and} \qquad \omega^0_2=B(z_1,z_2).
\]
% \label{eq:berg}
The kernel $K$ is defined by the following formula
\[
K(p_1,p)=\frac{-\int^p_{\hat{p}}\omega_2^0(p_1,p')}{2[y(p)-y(\hat{p})] \, \dd x(p)}%=\frac{1}{2[y(\hat{p})-y(p)] \, x'(p)}\left( \frac{1}{p-p_1}- \frac{1}{\hat{p}-p_1}\right)\frac{\dd p_1}{\dd p},
\] 
which is well-defined in the vicinity of each zero of $\dd x$. Note that the quotient of a differential by the differential $\dd x(p)$ is a meromorphic function.  For $2g-2+n>0$, the multidifferential $\omega^g_n$ is symmetric, with poles only at the zeros of $\dd x$ and vanishing residues.

The poles of the invariants $\omega^g_n$ occur at the zeros of $dx$ and are of order $6g-6+4n$ there.  One property of $\omega^g_n$ that we needed in Section~\ref{sec:locfac} and proven in \cite{EOrInv} is that the highest order part of $\omega^g_n$ near a pole is given by $y'(0)^{2-2g-n}\omega^{g\text{\ Airy}}_n$\hspace{-1mm}.  \hspace{1mm}Here $\omega^{g\text{\ Airy}}_n$ is the invariant obtained from the curve $y^2-x=0$ which is a local model near  zeros of $dx$ for any plane curve that is the input of topological recursion.  

Now that we have defined topological recursion precisely we need to qualify the 
$k=1$ part of \eqref{exactSk} which uses a  regularised version of \eqref{exactSk}.
\begin{equation}  \label{exactS01}
S_1(p)=\frac{1}{2!}\int^p\int^p\left[\omega^0_2(p_1,p_2)-\frac{dx_1dx_2}{(x_1-x_2)^2}\right]
\end{equation}
where $\omega^0_2(p_1,p_2)=B$ is part of the input data for the spectral curve.

\begin{remark}
Heuristically $\omega^{(g)}_n(p_1,p_2,...,p_n)=\langle{\rm Tr}\frac{1}{x(p_1)-A}...{\rm Tr}\frac{1}{x(p_n)-A}\rangle_c$ is the expected value of a resolvent in a matrix model.  The subscript $c$ means cumulant, or the connected part in a graphical expansion.  Topological recursion follows from the loop equations satisfied by the resolvents.
\end{remark}

For $2g-2+n>0$, the invariants $\omega^g_n$ of spectral curves satisfy the following {\em string equations} for $m = 0, 1$~\cite{EOrInv}.
\begin{equation}  \label{eq:string}
\sum_{\alpha} \res_{z=\alpha} x^my\omega^g_{n+1}(z,z_S)=-\sum_{j=1}^ndz_j\frac{\partial}{\partial z_j}\left(\frac{x^m(z_j)\omega^g_n(z_S)}{dx(z_j)}\right)
\end{equation}
They also satisfy the dilaton equation~\cite{EOrInv}
\begin{equation} \label{dilaton}
\sum_{\alpha}\res_{z=\alpha}\Phi(z)\, \omega^g_{n+1}(z,z_1, \ldots ,z_n)=(2-2g-n) \,\omega^g_n(z_1, \ldots, z_n),
\end{equation}
where the summation is over the zeros $\alpha$ of $\dd x$ and $\Phi(z)=\int^z y\,\dd x(z')$ is an arbitrary antiderivative. The dilaton equation enables the definition of the so-called {\em symplectic invariants}
\[
F_g=\sum_{\alpha}\res_{z=\alpha}\Phi(z)\,\omega^g_{1}(z).
\]
We will see in Section~\ref{sec:tfam} the importance of the string equations for the quantum curve.

\begin{remark}
Topological recursion generalises to {\em local} curves in which $C$ is an open subset of a compact Riemann surface.  This is because the recursive definition of $\omega^g_n(p_1,\ldots,p_n)$ uses only local information around zeros of $dx$.  Whereas, quantum curves are global in nature, essentially requiring the full algebraic structure of the curve.  Since the formula \eqref{exactSk} still makes sense for local curves, it suggests that one might be able to make sense of a local quantum curve.  This would be important because it was proven in \cite{DOSSIde} that any semi-simple cohomological field theory can be encoded via topological recursion applied to a local curve.  Thus it would raise the question of what role quantum curves can play in cohomological field theories.
\end{remark}

\subsection{Choice of primitive}
In \eqref{exactSk} the expression
$$S_k(p)=\sum_{2g-1+n=k}\int^p\int^p...\int^p\omega^g_n
$$
is ambiguous since integration is not uniquely defined.  It needs to be interpreted as follows.  A {\em primitive} $F^g_n(p_1,...,p_n)$ of $\omega^g_n(p_1,...,p_n)$ is a meromorphic function on the spectral curve $C$ that satisfies:
$$d_1...d_nF^g_n(p_1,...,p_n)=\omega^g_n(p_1,...,p_n)
$$
where $d_i$ is the exterior derivative in the $i$th coordinate $p_i$.  For example, choose $p_i,q_i\in C$, $i=1,...,n$, then one possible primitive of $\omega^g_n$ is given by $\int^{p_1}_{q_1}\int^{p_2}_{q_2}...\int^{p_n}_{q_n}\omega^g_n$ which is a function of $p_1,...,p_n$ with $q_i$ fixed.  With such a choice of primitive
$$S_k(p)=\sum_{2g-1+n=k}F^g_n(p,p,...,p)
$$
for $p\in C$.  The sum is finite hence $S_k(p)$ is a meromorphic function on the spectral curve.

The choice of primitive $F^g_n(p_1,...,p_n)$ is rather important.  For a function of a single variable one can retrieve the function from any primitive---simply differentiate---so the choice of primitive is not so important.  This remains true for a function of severable variables,  however it is no longer true for a {\em specialisation} of a primitive at a single variable.  Consider a function of severable variables $f(z_1,...,z_n)$ and any primitive $F(z_1,...,z_n)$, so $\frac{\partial}{\partial z_1}...\frac{\partial}{\partial z_n}F(z_1,...,z_n)=f(z_1,...,z_n)$.  Then there is no way to retrieve the function $f(z_1,...,z_n)$ from its specialisation $F(z,...,z)$.
\begin{example}
Consider $f(z_1,z_2)=3z_1^2+2z_2$ and $g(z_1,z_2)=4z_1z_2+2z_2$ and choose respective primitives $F(z_1,z_2)=z_1^3z_2+z_1z_2^2$ and $G(z_1,z_2)=z_1^2z_2^2+z_1z_2^2$. From the specialisation $F(z,z)=z^4+z^3=G(z,z)$ we cannot uniquely determine $f(z_1,z_2)$.
\end{example}

Usually a primitive is obtained by $S_k(p)=\sum_{2g-1+n=k}\int^p_{p_0}\int^p_{p_0}...\int^p_{p_0}\omega^g_n$, i.e. integration from a common point $p_0\in C$.  For rational curves there is another natural primitive.   On a rational curve $\omega^g_n$ is a sum of its principal parts.  Each principal part has its own natural primitive---take the principal part of any primitive.  Note that the specialisation $F^g_n(p,..,p)$ of a primitive of $\omega^g_n(p_1,..,p_n)$ obtained in this way from principal parts is in general not of the form
$\displaystyle\int^{p}_{p'}\int^{p}_{p'}...\int^{p}_{p'}\omega^g_n(p_1,..,p_n)$.

\subsubsection{A family of choices of primitives.}\label{sec:tfam}  The string equation \eqref{eq:string} allows us to generate a $t$-parametrised family of quantum curves via the action of the operator $e^{t\h\frac{d}{dx}}$ on wave functions.
\begin{proposition}  \label{th:famqc}
A quantum curve satisfying \eqref{exactSk} lives in a family of quantum curves.  Concretely, if $\psi(p,\h)$ satisfies \eqref{exactSk}, then the family of wave functions 
\begin{equation} \label{wavet}
\psi(p,\h,t)=e^{t\h\frac{d}{dx}}\psi(p,\h)
\end{equation}
also satisfy \eqref{exactSk} (up to unstable terms) for different choices of primitive.
\end{proposition}
\begin{proof}
The action of $e^{t\h\frac{d}{dx}}$ on $S_k(p)$ is given by
$$S_k(p,t)=\sum_{m=0}^kt^m\Big(\frac{d}{dx}\Big)^mS_{k-m}(p)$$
so in particular \eqref{wavet} is well-defined.

The new wave function also satisfies a wave equation, with the same semi-classical limit, simply by conjugation:
$$\Pw(\x,\y)\psi(p,\h)=0\Rightarrow e^{t\h\frac{d}{dx}}\Pw(\x,\y)e^{-t\h\frac{d}{dx}}\cdot \psi(p,\h,t)=0.$$
What is a little deeper is that $\psi(p,\h,t)$ also satisfies \eqref{exactSk} up to unstable terms.  The unstable terms are not important since one can conjugate by any discrepancy.

Define the linear functional $L_p$ on meromorphic functions on $C$ by
$$L_pf=\sum_{p=a_i}\res_{p=a_i}dy(p)f(p)$$
where the sum is over all points satisfying $dx(a_i)=0$.  Then 
\begin{equation} \label{dFstring}
\frac{d}{dx}F^g_n(p,p,\dots,p)=\sum_{j=1}^n\left.\frac{d}{dx(p_j)}F^g_n(p_1,\dots,p_n)\right|_{p_1=p_2=\dots=p}\hspace{-10mm}=L_pF^g_{n+1}(p,p_1,\dots,p_n)
\end{equation}
where the first equality is the chain rule and the second equality uses the string equation  \eqref{eq:string} with $m=0$.  We can iterate \eqref{dFstring} to make sense of higher derivatives $(\frac{d}{dx})^mF^g_n$.  

The idea is to adjust $F^g_{n+1}$ by terms of the type in the right hand side of \eqref{dFstring} to again get a symmetric function such as $F^g_{n+1}+t\sum_{j=1}^{n+1}L_{p_j}F^g_{n+1}(p_1,\dots,p_{n+1})$.  Define the function symmetric in $p_i$
\begin{equation} \label{Ft}
\cf^g_n(t,p_1,\dots,p_n)=\sum_{m=0}^k\sum_{\{i_1,...,i_m\}}t^mL_{p_{i_1}}\cdots L_{p_{i_m}}F^g_n(p_1,\dots,p_n)
\end{equation}
where the second sum is over all cardinality $m$ subsets of $\{1,...,n\}$.  Notice that
$$ d_{p_1}\cdots d_{p_n}\cf^g_n(t,p_1,\dots,p_n)=d_{p_1}\cdots d_{p_n}F^g_n(p_1,\dots,p_n)=\omega^g_n(p_1,\dots,p_n)$$
since we have adjusted only by summands independent of at least one variable $p_j$ and hence annihilated by $d_{p_1}\cdots d_{p_n}$.

Hence we see that $S_k(p,t)$ satisfies \eqref{exactSk} for the choice of primitive $\cf^g_n(t,p_1,\dots,p_n)$ and the proposition is proven.
\end{proof}
The proposition is interesting both for the extra structure it brings to the wave functions---essentially a relationship with the string equation---and to emphasise the fact that there there is a choice of primitive.  We see that \eqref{exactSk} still leaves some ambiguity in the construction of the quantum curve.  
\begin{remark}
The main tool in the proof of Proposition~\ref{th:famqc} is the string equation \eqref{eq:string} that comes out of topological recursion.  In Section~\ref{sec:ex} we apply this idea to the quantum curve associated to Gromov-Witten invariants of the sphere and we find that the string equation applied there coincides with the string equation that comes out of Gromov-Witten invariants.  Furthermore, we see that the $t$-parametrised family $\cf^g_n(t,p_1,\dots,p_n)$ used there has enumerative meaning---it is related to insertions of the so called puncture operator, and gives a relation to the Toda equations. 
\end{remark}

\section{Enumerative examples}   \label{sec:ex}

Many examples of the conjectural relation between quantum curves and topological recursion have been proven in the literature.  The quantum curve was shown to satisfy \eqref{exactSk} for simple Hurwitz numbers \cite{MSuSpe,ZhoQua} and simple Hurwitz numbers of an arbitrary base curve \cite{LMSQua};  monotone Hurwitz numbers \cite{DDMTop}; Belyi maps \cite{MSuSpe}; bipartite Belyi maps \cite{DNoTop,KZoVir}; hypermaps \cite{DMaQua}; Gromov-Witten invariants with target $X$ for $X=\{\text{pt}\}$ \cite{ZhoInt}; $X=\bp^1$ \cite{DMNPSQua}; $X=$ the topological vertex and the resolved conifold \cite{ZhoQua}; and spectral curves coming from matrix models \cite{BBERat}.

There are essentially two types of proofs.  Firstly there are the proofs that use the $\h$ expansion of $\log\psi(p,\h)$ and stay closer to the conjecture, essentially giving a reason for the conjecture.  Secondly there are the proofs that use an expansion of $\psi(p,\h)$ in $x$ and exploit the fact that the Euler characteristic is a coarser invariant than genus, and hence the proofs may be simpler.   It is rather natural to use an $x$ expansion when implementing the quantum curve on a computer. %Thirdly there are the proofs that use an exact formula for $\psi(p,\h)$.  The second and third approaches are sometimes closely related.  
One might also describe the $\h$ expansion proof as a proof concentrated around the zeros of $dx$, while the second type of proof considers expansions at regular values of $x$.  The quantum curve for simple Hurwitz numbers and Belyi maps are examples that have been proven using both approaches---the first approach in \cite{MSuSpe} and the second approach in \cite{ZhoInt}, respectively \cite{DMaQua}.  

In the remainder of this section we describe two specific examples of quantum curves.  In the first of these---the quantum curve of the Gromov-Witten invariants of $\bp^1$---we consider a family of quantum curves parametrised by $t\in\bc$ with semi-classical limit independent of $t$.  They correspond to different choices of primitives in \eqref{exactSk}.  In the second of these---Belyi maps---we present two proofs in order to contrast the two approaches described above used in most proofs.

\subsection{Gromov-Witten invariants of $\bp^1$}
In the following example the wave function is given by a specialisation of the partition function of Gromov-Witten invariants of $\bp^1$.  We consider a family of quantum curves parametrised by $t\in\bc$ with semi-classical limit independent of $t$. The $t=1/2$ case appeared in \cite{DMNPSQua}.  The $t$-dependence is rather useful, making a connection with the Toda lattice.  

Let $\Mbar_{g,n}(\bp^1,d)$ denote the moduli space of stable maps of degree $d$ from an $n$-pointed genus $g$ curve to $\bp^1$. The  descendant Gromov-Witten invariants of $\bp^1$ are defined by 
\begin{equation}
\label{eq:GW}
\langle \prod_{i=1}^n
\tau_{b_i}(\alpha_i)\rangle_{g.n} ^d
:=
\int _{[\Mbar_{g,n}(\bp^1,d)]^{vir}}
\prod_{i=1}^n \psi_i ^{b_i} ev_i^*(\alpha_i),
\end{equation}
where 
$[\Mbar_{g,n}(\bp^1,d)]^{vir}$ is the virtual 
fundamental class of the moduli space, of degree given by its virtual dimension $2g-2+n +2d$,
\begin{equation*}
ev_i:\Mbar_{g,n}(\bp^1,d)\longrightarrow \bp^1
\end{equation*}
is a natural morphism defined by evaluating a stable map at the $i$-th marked point of the source curve, $\alpha_i\in H^*(\bp^1,\bq)$ is a cohomology class of the target $\bp^1$, and $\psi_i$ is  the tautological cotangent class in $H^2(\Mbar_{g,n}(\bp^1,d),\bq)$.  We denote by $1$ the generator of $H^0(\bp^1,\bq)$, and by $\omega\in H^2(\bp^1,\bq)$ the Poincar\'e dual to the point class.  We call $\tau_k(\omega)$ {\em stationary} classes since the pull-back $ev_i^\ast(\omega)\subset \modm_{g,n}(\bp^1,d)$ restricts to stable maps $f$ with $f(p_i)=x_i$ for a given stationary point $x_i\in\bp^1$.  

The free energy of the Gromov-Witten invariants of $\bp^1$ is defined by
$$ F_g=  \sum_dq^d\left\langle\exp\left\{\sum_{i\geq 0}^{\infty}\tau_i(\omega)t_i+\tau_i(1)s_i\right\}\right\rangle^g_d$$
and
$$
F=\sum_{g\geq 0} F_g=\frac{1}{2}s_0^2t_0+\frac{1}{6}s_0^3t_1+...-\frac{1}{24}t_0-\frac{1}{24}s_0t_1+...
+q(1+t_0+\frac{1}{2}t_0^2+...+s_0t_1+\frac{1}{2}s_0^2t_2+...)+...\quad .
$$
The partition function is defined to be
$$Z(t_0,t_1,...,s_0,s_1,...,q)=\exp F.$$

\subsubsection{Quantum curve}
A specialisation of the partition function gives rise to a wave function and quantum curve associated to Gromov-Witten invariants of $\bp^1$.  We will describe its semi-classical limit below.  Define
$$
\psi(x,\hbar,q,t)=Z\Big(t_i=i!\left(\frac{\hbar}{x}\right)^{i+1}\hspace{-.3cm},q=\frac{q}{\hbar^2},s_0=t,s_i=0,i>0\Big)
=\exp\left\{\frac{q}{\hbar^2}+\frac{-\frac{\hbar}{24}+\frac{q}{\hbar}+\frac{1}{2}t^2\hbar}{x}+\frac{..}{x^2}+...\right\}$$
Note that we have switched off all the non-stationary insertions except the class $\tau_0(1)$ which is known as the {\em puncture operator}.   Geometrically we have placed the target stationary points at a single point $p\in\bp^1$.  The coefficient of $x^k$ counts all covers with local {\em virtual} degree $k$ over $p$ defined to be the total ramification plus number of preimage points over $p$.  This local virtual degree can differ from the actual local degree.  For example, the descendant $\tau_1(\omega)$ can be realised on the space of degree 1 stable maps $\Mbar_{1,1}(\bp^1,1)$ hence appears as $\langle\tau_1(\omega)\rangle x^{-2}$ in $\psi$.  The exponent $-2$ of $x$, which corresponds to local virtual degree 2, reflects that fact that $\tau_1(\omega)$ should arise from a locally 2 to 1 map.

The wave function of the quantum curve is obtained from $\psi$ by modifying the unstable terms:
\begin{equation}  \label{unstwave}
\psi_1(x,\hbar,q,t)=\psi(x,\hbar,q,t)\cdot x^{-t}\exp(\frac{1}{\hbar}(x\ln x-x)).
\end{equation}
\begin{theorem} \cite{DMNPSQua} \label{th:QGW}
$$
[e^{\hbar\frac{d}{dx}}+qe^{-\hbar\frac{d}{dx}}-x+(t-\frac{1}{2})\hbar]\psi_1=0.$$
\end{theorem}
\begin{remark}
The proof in \cite{DMNPSQua} considers the case $q=1$, $t=1/2$ but it is not difficult to derive the statement here from that special case.
\end{remark}
\begin{remark}
The operator in Theorem~\ref{th:QGW} is a simple case of the Lax operator for the Toda lattice appearing in Dubrovin-Zhang and Takasaki-Takebe:
$$e^{\hbar\frac{d}{dx}}+v(x)+e^{u(x)}e^{-\hbar\frac{d}{dx}}.$$
It appears in Aganagic-Dijkgraaf-Klemm-Mari\~no \cite{ADKMVTop} as
$$H=e^{\hbar\frac{d}{dx}}+x+e^{-\hbar\frac{d}{dx}}$$
and the wave function $\psi$ is said to describe the insertion of a $D$-brane at a fixed $x$.  
\end{remark}

The semi-classical limit gives rise to the following spectral curve.  Put $\psi_1=\exp(\frac{1}{\hbar}S_0+S_1+\hbar S_2+...)$.  
\begin{align*}
0=\lim_{\hbar\to 0}\exp^{-\frac{1}{\hbar}S_0}&\left[e^{\hbar\frac{d}{d x}}+qe^{-\hbar\frac{d}{d x}}-x+(t-\frac{1}{2})\hbar\right]\exp(\frac{1}{\hbar}S_0+S_1+\hbar S_2+...)\\
&=\lim_{\hbar\to 0}[e^{\frac{S_0(x+\hbar)-S_0(x)}{\hbar}}e^{\hbar\frac{d}{d x}}+qe^{\frac{S_0(x-\hbar)-S_0(x)}{\hbar}}e^{-\hbar\frac{d}{d x}}-x+(t-\tfrac{1}{2})\hbar]\exp(S_1+\hbar S_2+...)\\
&=[e^{S_0'(x)}+qe^{-S_0'(x)}-x]\exp(S_1).
\end{align*}
Hence $e^{S_0'(x)}+qe^{-S_0'(x)}-x=0$ and for $z=e^{S_0'(x)}$ this defines the spectral curve
\begin{equation}  \label{specGW}
C=\{x=z+\frac{q}{z},\quad y=\ln z\}
\end{equation}
which agrees with the mirror Landau-Ginzburg model \cite{EHYTop}.  Although it is not algebraic, one can still apply topological recursion to $C$ to produce the stationary Gromov-Witten invariants of $\bp^1$.
\begin{theorem}[\cite{DOSSIde},\cite{NScGro}]  \label{th:GWEO}
(Analytic expansions around a branch of $\{x_i=\infty\}$ of) the invariants $\omega^g_n(C)$ of the curve $C$ defined in (\ref{specGW}) are generating functions for the stationary Gromov-Witten invariants of $\bp^1$:
\[ \omega^g_n=\sum_{\bf b}\left\langle \prod_{i=1}^n \tau_{b_i}(\omega) \right\rangle^g_d\cdot\prod_{i=1}^n(b_i+1)!x_i^{-b_i-2}dx_i-\delta_{g0}\delta_{n1}\ln{x_1}dx_1+\delta_{g0}\delta_{n2}\frac{dx_1dx_2}{(x_1-x_2)^2}.\]
\end{theorem}
\begin{remark} The non-stationary Gromov-Witten invariants of $\bp^1$ are also contained inside the $\omega^g_n(C)$ in the form of ancestor invariants.  See \cite{DOSSIde} for details.
\end{remark}
Indeed, Theorems~\ref{th:QGW} and \ref{th:GWEO} confirm the conjectural form \eqref{exactSk} that $S_k(p)=F^g_n(p,p,...,p)$ for $F^g_n(p_1,...,p_n)$ a primitive of $\omega^g_n$ of the spectral curve \eqref{specGW}.

\subsubsection{Relation to Toda equation}
It is useful to replace the derivatives in $x$ with derivatives in the parameter $t$ using the wave operator
$$D:=\hbar\frac{d}{d x}+\frac{\partial}{\partial t}$$
and the equation
$$D\psi_1=0$$
which follows from the string equation (satisfied quite generally by Gromov-Witten invariants) restricted to the specialisation of the partition function.  The quantum curve becomes
\begin{equation}  \label{QDEt}
\psi(t-1)+\frac{q}{x^2}\psi(t+1)-\Big(1-(t-\frac{1}{2})\frac{\hbar}{x}\Big)\psi(t)=0
\end{equation}
where $\psi(t)$ is defined via \eqref{unstwave} and we suppress all arguments except for $t$, so for example we write $\psi(t-1)$ for $\psi(x,\hbar,q,t-1)$.

The partition function satisfies the Toda lattice equation 
\[\frac{Z(s_0-1)Z(s_0+1)}{Z(s_0)^2}=\frac{1}{q}\frac{\partial^2}{\partial t_0^2}\log Z(s_0)\]
which was conjectured by Eguchi-Yang \cite{EYaTop} and proven by Okounkov-Pandharipande \cite{OPaGro} for $s_0=n\in\bz$ and Dubrovin-Zhang \cite{DZhVir} for $s_0\in\br$.   Using the divisor equation
$$\frac{\partial F}{\partial t_0}=\frac{1}{2}s_0^2+q\frac{\partial F}{\partial q}$$
one can restrict the Toda lattice equation to the specialisation to get:
\[\frac{\psi(t-1)\psi(t+1)}{\psi(t)^2}=\h^2\frac{d}{dq}\left(q\frac{d}{dq}\log \psi\right).\]
The Toda lattice equation does not uniquely determine the partition function $Z$, nor does its specialisation determine the wave function $\psi$.  Together with $\psi(x,q=0,\hbar,t)$ the Toda lattice equation
does determine $\psi(x,q,\hbar,t)$.

Write \eqref{QDEt} as $\cl\psi=0$, so
$$\cl=e^{-\hbar\frac{\partial}{\partial t}}+\frac{q}{x^2}e^{\hbar\frac{\partial}{\partial t}}-(1-(t-\frac{1}{2})\frac{\hbar}{x}).$$
Then
$$(D-(t-1)\frac{\hbar}{x})\circ\cl=\cl\circ(D-t\frac{\hbar}{x})$$
and $\psi$ is characterised as the solution of \eqref{QDEt} that is also an eigenfunction of $D$:
$$D\psi(x,\hbar,q,t)=t\frac{\hbar}{x}\psi(x,\hbar,q,t).$$
Similarly, $D$ is compatible with the Toda lattice equation so it admits an eigenfunction solution if the degree 0 part is an eigenfunction of $\cl$.  Note that there are solutions of Toda that are not solutions of \eqref{QDEt}  and vice versa.

The {\em degree zero} Gromov-Witten invariants contrast nicely the difference between the Toda lattice equation and the quantum curve.  Put $\psi=\sum_{d\geq 0}q^d\psi_d=\psi_0+q\psi_1+...$

\underline{Toda equation}$|_{q=0}$:\quad $\displaystyle\frac{\psi_0(t+1)\psi_0(t-1)}{\psi_0(t)^2}=\hbar^2\psi_1(t)$ has nothing to say about $\psi_0$.

\underline{Quantum curve}$|_{q=0}$: $\psi_0(t-1)=(1-(t-\frac{1}{2})\frac{\hbar}{x})\psi_0(t)=0$ uniquely determines $\psi_0(x,\hbar,t)$ (which is also an eigenfunction of $D$).  An exact formula \cite{OPaGro,FPaHod} is given by $\psi_0=\exp{F_0}$ where
$$F_0=\phi(x-\hbar t) +\frac{1}{\hbar}(x-\hbar t)\ln\left(1-t\frac{\hbar}{x}\right)+t,\quad \phi(x)=\sum_{g=1}^{\infty}(1-2^{1-2g})\frac{\zeta(1-2g)}{2g-1}\left(\frac{\hbar}{x}\right)^{2g-1}.
$$
Finally, the quantum curve implies rather nice rational behaviour of the wave function.  Applying the quantum curve equation \eqref{QDEt} to $\psi=\sum_{d\geq 0}q^d\psi_d=\psi_0\left(1+q\frac{\psi_1}{\psi_0}+q^2\frac{\psi_2}{\psi_0}+...\right)$, one finds that $\frac{\psi_d}{\psi_0}$ is
{\em rational} in $x$ and with $d$ simple poles.  For example,
$$\frac{\psi_1}{\psi_0}(x,\hbar,t)=1+\frac{1}{\frac{x}{h}-t-\frac{1}{2}},\ 
\frac{\psi_2}{\psi_0}(x,\hbar,t)=\frac{1}{2}+\frac{\frac{1}{2}}{\frac{x}{h}-t-\frac{1}{2}}+\frac{\frac{1}{2}}{\frac{x}{h}-t-\frac{3}{2}}.
$$
Put $w=\frac{x}{\hbar}-t$ and
$$\frac{\psi_d}{\psi_0}=r_d(w)=a_{0,d}+\frac{a_{1,d}}{w-\tfrac{1}{2}}+\frac{a_{2,d}}{w-\tfrac{3}{2}}+...+\frac{a_{d,d}}{w-d+\tfrac{1}{2}}$$
then the residues satisfy the linear system: 
\begin{align*}
a_{i,d}&=a_{i+1,d}+\frac{a_{i-1,d-1}}{i(i-1)}\quad (a_{d+1,d}=0)\\
a_{1,d}&=a_{2,d}+r_{d-1}(-\tfrac{1}{2}).
\end{align*}

\subsection{Belyi maps}

The fundamental example of the spectral curve $y^2-xy+1=0$ corresponds to the enumeration of Belyi maps.  We present two proofs that its quantum curve satisfies \eqref{exactSk}.  The first proof uses the $\h$ expansion of (the log of) $\psi(p,\h)$ while the second proof uses the expansion of $\psi(p,\h)$ around $x=\infty$.  The first proof uses topological recursion and thus in a sense explains why the conjecture is  true.  The second proof uses the underlying enumerative problem and highlights a connection to hypergeometric functions.

Let $\cb_{g,n}(\mu_1, \ldots, \mu_n)$ be the set of all connected genus $g$ {\em Belyi} maps---meaning branched covers $\pi:\Sigma\to\bp^1$ unramified over $\bp^1\backslash\{0,1,\infty\}$---with all points over 1 having ramification 2 and ramification divisor over $\infty$ given by $\pi^{-1}(\infty)=\mu_1p_1+ \cdots + \mu_np_n$, where the points over $\infty$ are labelled $p_1, \ldots, p_n$.  Two Belyi maps $\pi_1:\Sigma_1\to\bp^1$ and $\pi_2:\Sigma_2\to\bp^1$ are isomorphic if there exists a homeomorphism $f:\Sigma_1\to\Sigma_2$ that covers the identity on $\bp^1$ and preserves the labelling over $\infty$. 
\begin{definition} 
For any $\mmu=(\mu_1,\ldots,\mu_n)\in\bz_+^n$, define
\[
M_{g,n}(\mu_1,\ldots,\mu_n)=\sum_{\pi\in\cb_{g,n}(\mmu)}\frac{1}{|{\rm Aut\ }\pi|},
\]
where ${\rm Aut\ }\pi$ denotes the automorphism group of the branched cover $\pi$.
\end{definition}
Define
\begin{equation}  \label{genbelyi}
F^g_n(x_1,...,x_n)=(-1)^n\sum_{\mmu> 0}M_{g,n}(\mu_1, \ldots, \mu_n)x_1^{-\mu_1} \ldots x_n^{-\mu_n}+\delta_{g0}\delta_{n1}\log x.
\end{equation}
The exceptional case of $(g,n)=(0,1)$ includes an extra $\log x$ term (to allow for the 0th Catalan number $C_0=1$ missing from $kM_{0,1}(k)=C_k$) so that
$$
\frac{d}{dx}F^0_1(x)=\sum_{k\geq 0}C_kx^{-k-1}=y,\quad x=y+\frac{1}{y}.
$$
This is the same as Example~\ref{ex2}, and gives an analytic expansion of the global meromorphic function $y$ in the local coordinate $x$ on a branch of $x = \infty$, of the Stieltjes transform of a probability measure
$$y=\hat{\rho}(x)=\frac{1}{2\pi}\int^2_{-2}\frac{\sqrt{4-t^2}}{x-t}dt.$$
The connection between the two goes deeper.  The probability measure is the Wigner semicircle distribution of eigenvalues of Hermitian matrices with the Gaussian potential.  Associated to any Belyi map $f\in\cb_{g,n}(\mu_1, \ldots, \mu_n)$ is a fatgraph given by the pull-back of the unit interval $\Gamma=f^{-1}([0,1])$, and these graphs arise when calculating Hermitian matrix integrals.  The {\em spectral curve} of this matrix model is
\begin{equation}  \label{specbelyi}
y^2-xy+1=0.
\end{equation}
It was proven in \cite{EOrTop} that Tutte's equations for discrete surfaces---the fatgraphs $\Gamma=f^{-1}([0,1])$---correspond to the loop equations for this matrix model:
\begin{align}  \label{loop}
x_1W_g(x_1,\xx_S) =& W_{g-1}(x_1,x_1,\xx_S)+ \mathop{\sum_{g_1+g_2=g}}_{I \sqcup J = S} W_{g_1}(x_1,\xx_I) \, W_{g_2}(x_1,\xx_J)\\
&\quad+\sum_{j=2}^n \frac{\partial}{\partial x_j}\frac{W_g(x_1,\xx_{S\setminus\{j\}})-W_g(\xx_S)}{x_1-x_j}+\delta_{g,0} \, \delta_{n,1}
\nonumber %%%
\end{align}
where $W_g$ is related to the generating function \eqref{genbelyi} via
$$\frac{\partial}{\partial x_1} \cdots \frac{\partial}{\partial x_n} F_{g,n}(x_1, \ldots, x_n)=W_g(x_1, \ldots, x_n).
$$
The solution of the loop equations \eqref{loop} for $(g,n)=(0,1)$ defines the spectral curve \eqref{specbelyi}.
A consequence is that topological recursion applied to the spectral curve \eqref{specbelyi} yields:
\begin{equation}  \label{TRbelyi}
\omega^g_n(p_1, \ldots, p_n) = W_g(x_1, \ldots, x_n)\, \dd x_1 \otimes \cdots \otimes \dd x_n-\delta_{g0}\delta_{n1}\frac{dx_1}{x_1}+\delta_{g0}\delta_{n2}\frac{dx_1\otimes dx_2}{(x_1-x_2)^2}
\end{equation}
for $p_i\in C$, the spectral curve, $x_i=x(p_i)$ and equality denotes an analytic expansion in the local coordinate $x$ on a branch of $x = \infty$.  We see that the unstable cases $(g,n)=(0,1)$ and $(0,2)$ require minor adjustments, as in the example of Gromov-Witten invariants of $\bp^1$.  

Consider the following wave function constructed out of $F^g_n(x_1,...,x_n)$ defined in \eqref{genbelyi} for $x_i=x(y_i)=y_i+\frac{1}{y_i}$ and specialised to $y_i=y$.
\begin{equation}  \label{wavefun} 
\psi(y,\h)=\exp\Big(\h^{-1}\sum_{k\geq 0} \h^kS_k(y)\Big),\quad S_k(y)=\sum_{2g-1+n=k}\frac{1}{n!}F^g_n(x(y),x(y),\ldots,x(y))
\end{equation}
where $x(y)=y+\frac{1}{y}$.
The formulae \eqref{TRbelyi} show that $\psi(y,\h)$ is the the conjectural wave function \eqref{exactSk} for the quantum curve of the spectral curve \eqref{specbelyi}.  The following theorem confirms the conjectural form.  
\begin{theorem}\cite{GSuApo,MSuSpe}  \label{thm:qcurve}
\begin{equation}  \label{qde}
\left(\h^2\frac{d^2}{dx^2}-x\h\frac{d}{dx}+1\right)\psi(y,\h)=0
\end{equation}
for $y^2-xy+1=0$ and $\psi(y,\h)$ defined in \eqref{wavefun}.
\end{theorem}
\begin{remark}
This theorem was known in the physics literature, see for example \cite{GSuApo}.  A rigorous proof using topological recursion was given in \cite{MSuSpe}.  A simpler more direct combinatorial proof was given in \cite{DMaQua}.  
\end{remark}

\begin{proof}

{\bf Proof 1---expansion in $\h^k$}  following \cite{MSuSpe}.

The unstable cases have been proven, since for $(g,n)=(0,1)$, $\frac{d}{dx}S_0(z)=z$ is shown above and for $(g,n)=(0,2)$, it follows from the relation
$$d_1d_2F^0_2(x(z_1),x(z_2))=\frac{dz_1dz_2}{(z_1-z_2)^2}-\frac{dx(z_1)dx(z_2)}{(x(z_1)-x(z_2))^2}=\frac{dz_1dz_2}{(1-z_1z_2)^2}\Rightarrow S_1(z)=-\frac{1}{2}\log(1-z^2)$$
which agrees with $S_1$ calculated via the WKB method in \eqref{ex2s1} with $y=1/z$.

%We deal first with the unstable cases.  As shown above $\dfrac{d}{dx}S_0(z)=z=y$ and
%$$F_{0,2}(z)=\int_0^z\int_0^z\left(\frac{dz_1dz_2}{(z_1-z_2)^2}-\frac{dx(z_1)dx(z_2)}{(x(z_1)-x(z_2))^2}\right)=\int_0^z\int_0^z\frac{dz_1dz_2}{(1-z_1z_2)^2}=-\log(1-z^2)$$
%hence
%$$ \frac{1}{2!}\frac{d}{dx}F_{0,2}(z)=\frac{1}{2!}\frac{1}{x'(z)}\frac{d}{dz}F_{0,2}(z)=-\frac{z^3}{(z^2-1)^2}=-\frac{\left(\frac{d}{dx}\right)^2S_0(z)}{2\frac{d}{dx}S_0(z)-x}=\frac{d}{dx}S_1(z)$$
%as required.

For $\log\psi(z,\h)=\h^{-1}\sum_{k\geq 0} \h^kS_k(z)$ \eqref{qde} becomes
$$
\h^2\left(\left(\frac{d}{dx}\right)^2\log\psi(z,\h)+\left(\frac{d}{dx}\log\psi(z,\h)\right)^2\right)-x\h\frac{d}{dx}\log\psi(z,\h)+1=0
$$
which is equivalent to
\begin{equation}\label{qloop}
\sum_{k=0}^{\infty}\left(\frac{d}{dx}\right)^2S_k(z)\h^{k+1}+\left(\sum_{k=0}^{\infty}\frac{d}{dx}S_k(z)\h^k\right)^2-x\sum_{k=0}^{\infty}\frac{d}{dx}S_k(z)\h^k+1=0.
\end{equation}
Note that if we set $\h=0$ in \eqref{qloop} we get $\left(\frac{d}{dx}S_0(z)\right)^2-x\frac{d}{dx}S_0(z)+1=0$ which is exactly the first of the loop equations \eqref{loop}.
The essential idea of the proof is that the coefficients of $\h^k$ for $k>0$ have the quadratic form of the topological recursion.  

Integrate \eqref{loop} with respect to $x_2,...,x_n$ (when $n>1$):
\begin{align}  \label{iloop}
x_1\frac{d}{dx_1}F^g_n(x_1,\xx_S) =& \left.\frac{\partial^2}{\partial u_1\partial u_2}F^{g-1}_{n+1}(u_1,u_2,\xx_S)\right|_{u_1=u_2=x_1}\hspace{-1mm}+\hspace{-1mm} \mathop{\sum_{g_1+g_2=g}}_{I \sqcup J = S} \frac{d}{dx_1}F^{g_1}_{|I|+1}(x_1,\xx_I) \, \frac{d}{dx_1}F^{g_2}_{|J|+1}(x_1,\xx_J)\\
&+\sum_{j=2}^n \frac{\frac{d}{dx_1}F^g_{n-1}(x_1,\xx_{S\setminus\{j\}})-\frac{d}{dx_j}F^g_{n-1}(\xx_S)}{x_1-x_j}.
\nonumber %%%
\end{align}
Each term of \eqref{iloop} vanishes at $x_j=\infty$  for any $j=2,...,n$ which determines the constants of integration.

Specialise \eqref{iloop} to $x_i=x$
\begin{align}  \label{sloop}
\frac{1}{n}x\frac{d}{dx}F^g_n(x,\dots,x) =& \mathop{\sum_{g_1+g_2=g}} \binom{n-1}{|I|}\frac{1}{|I|+1}\frac{d}{dx}F^{g_1}_{|I|+1}(x,\dots,x) \, \frac{1}{|J|+1}\frac{d}{dx}F^{g_2}_{|J|+1}(x,\dots,x)\\
+\frac{1}{n(n+1)}\frac{d^2}{dx^2}&F^{g-1}_{n+1}(x,\dots,x)-\frac{1}{n}\left.\frac{d^2}{du^2}F^{g-1}_{n+1}(u,x,\dots,x)\right|_{u=x}
\hspace{-5mm}+(n-1)\left.\frac{d^2}{du^2}F^g_{n-1}(u,x,\dots,x)\right|_{u=x}
\nonumber %%%
\end{align}
which uses the elementary relations on symmetric polynomials $\left.\frac{d}{dx_1}F^g_n(x_1,x,\dots,x)\right|_{x_1=x}=\frac{1}{n}\frac{d}{dx}F^g_n(x,\dots,x)$ and $\left.\frac{\partial^2}{\partial u_1\partial u_2}F^{g-1}_{n+1}(u_1,u_2,x,\dots,x)\right|_{u_i=x}\hspace{-2mm}=\frac{1}{n(n+1)}\frac{d^2}{dx^2}F^{g-1}_{n+1}(x,\dots,x)-\frac{1}{n}\left.\frac{d^2}{du^2}F^{g-1}_{n+1}(u,x,\dots,x)\right|_{u=x}$ and the fact that the limit $x_1\to x_j$ in the last term defines the derivative.  See Appendix A in \cite{MSuSpe}.

Recall that $\displaystyle S_k(z)=\sum_{2g-1+n=k}\frac{1}{n!}F^g_n(x(z),x(z),\ldots,x(z))$.  So we multiply \eqref{sloop} by $\frac{1}{(n-1)!}$ and sum over all $(g,n)$ such that $2g-1+n=k$.  The second derivatives in $u$ cancel and \eqref{sloop} becomes
$$
x\frac{d}{dx}S_k(z) = \sum_{i+j=k} \frac{d}{dx}S_i(z) \frac{d}{dx}S_j(z)+\frac{d^2}{dx^2}S_{k-1}(z)
$$
which is the coefficient of $\h^k$ in \eqref{qloop} for $k>0$.  Hence the theorem is proven.\\

{\bf Proof 2---expansion in $x^{-1}$} following \cite{DMaQua}.  First express the wave function defined in \eqref{wavefun} using \eqref{genbelyi} as
\[
\psi(p, \h) = x^{1/\h} \, \opsi(x, \h).
\]
%where the term $x^{-1/\h}$ comes from the exceptional logarithmic term in the definition of $F_{0,1}$.
  We write $x=x(p)$ in the argument because we will work with expansions around $x=\infty$.  %We first convert the series in $\h^{\pm1}$ to a Laurent series in $\h$ In order to interpret Theorem~\ref{thm:qcurve}, we need to make precise what we mean by this equation, given that the $\h$-expansion of $Z(x, \h)$ is not well-defined. One way to do this is to 
The main idea is to prove the following exact formula.  We will see later its relation to Hermite polynomials.
\begin{equation}  \label{xexpansion}
\opsi(x, \h) = 1 + \sum_{e=1}^\infty \frac{ (-1)^{e}\h^e}{2^ee!}  \h^{-1} (\h^{-1} - 1) (\h^{-1} - 2) \cdots (\h^{-1} - 2e+1) x^{-2e}.
\end{equation}
To prove the wave equation \eqref{qde} we simply apply the differential operator directly to the formula \eqref{xexpansion}.  Actually, we would like to work with $\opsi(x, \h)$ which has an expansion in $x^{-1}$ with coefficients that are Laurent polynomials in $\h$~---~in other words, $\opsi(x, \h) \in \mathbb{Q}[\h^{\pm 1}][[x^{-1}]]$.  Thus we first conjugate \eqref{qde} by $x^{1/\h}$ to get
%$$x^{-1/\h}\left(\h^2\frac{d^2}{dx^2}-x\h\frac{d}{dx}+1\right)x^{1/\h}\opsi(x, \h)=0$$ hence
\begin{equation} \label{abc}
\left[ \h^2 \frac{d^2}{d x^2} + \h \left((\frac{2}{x}-x)\frac{d}{d x}-\frac{1}{x^2}\right)  +\frac{1}{x^2}\right] \opsi(x, \h) = 0.
\end{equation}
which follows immediately by direct application to the formula \eqref{xexpansion}. \\
\\
The remainder of the proof is devoted to proving \eqref{xexpansion}.
First, consider the logarithm of the modified wave function.
\begin{align*}
\log \opsi(x, \h) &= \sum_{g=0}^\infty \sum_{n=1}^\infty \frac{\h^{2g-2+n}}{n!} \, F_{g,n}(x, x, \ldots, x) 
= \sum_{g=0}^\infty \sum_{n=1}^\infty \frac{\h^{2g-2+n}}{n!} \hspace{-4mm}\sum_{\mu_1, \ldots, \mu_n = 1}^\infty\hspace{-4mm} (-1)^nM_{g,n}(\mu_1, \ldots, \mu_n) \, x^{-\sum\mu_i } \\
&= \sum_{v=1}^\infty \sum_{e=1}^\infty (-1)^{e-v}f(v,e) \, \h^{e-v} x^{-2e}
\end{align*}
Here, $f(v, e)$ denotes the weighted count of connected dessins with $v$ vertices and $e$ edges. To obtain this last expression, we have used the fact that $e - v = 2g - 2 + n$ and $\mu_1 + \cdots + \mu_n = 2e$ for any dessin. The factor $\frac{1}{n!}$ accounts for the fact that we are now considering dessins with unlabelled faces. The weight of a dessin is the reciprocal of its number of automorphisms.  %Note that we exclude from consideration the dessins consisting of an isolated vertex.

Let $f^\bullet(v, e)$ denote the weighted count of possibly disconnected dessins with $v$ vertices and $e$ edges. Then
\[
\opsi(x, \h) = 1 + \sum_{v=1}^\infty \sum_{e=1}^\infty (-1)^{e-v} f^\bullet(v, e) \, \h^{e-v} x^{-2e}.
\]
Now $f^\bullet(v, e)$ is equal to $\frac{1}{(2e)!}$ multiplied by the number of triples $(\sigma_0, \sigma_1, \sigma_2)$ of permutations in the symmetric group $S_{2e}$ such that $\sigma_0 \sigma_1 \sigma_2 = \text{id}$, $\sigma_0$ has $v$ disjoint cycles and $\sigma_1$ has cycle type $2^e$.  Clearly we need to sum only over pairs $(\sigma_0, \sigma_1)$. Recall that the Stirling number of the first kind $\stirling{n}{k}$ counts the number of permutations in $S_n$ with $k$ disjoint cycles.  Since $(2e-1)!!$ is the number of permutations in $S_{e}$ of cycle type $2^e$ we have
\[
f^\bullet(v, e) = \frac{1}{(2e)!} \stirling{2e}{v} (2e-1)!!= \frac{1}{2^ee!} \stirling{2e}{v} 
\]
Note that $\opsi(x, \h)\in\mathbb{Q}[\h^{\pm 1}][[x^{-1}]]$ since for fixed $e$ we require $2 \leq v \leq 2e$ to have $f^\bullet(v, e) \neq 0$.

Now we simply use the fact that the generating function for Stirling numbers of the first kind is given by
\[
\sum_{k=1}^n \stirling{n}{k} x^k = x (x+1) (x+2) \cdots (x+n-1).
\]
Use this in the expression for the modified wave function as follows.
\begin{align*}
\opsi(x, \h) &= 1 + \sum_{v=1}^\infty \sum_{e=1}^\infty \frac{ (-1)^{e-v}}{2^ee!}  \stirling{2e}{v} \h^{e-v} x^{-2e} 
= 1 + \sum_{e=1}^\infty   (-1)^{e}\h^ex^{-2e}\frac{1}{2^ee!}\sum_{v=1}^\infty   \stirling{2e}{v} (-\h)^{-v} \\
&= 1 + \sum_{e=1}^\infty \frac{ (-1)^{e}\h^e}{2^ee!}  \h^{-1} (\h^{-1} - 1) (\h^{-1} - 2) \cdots (\h^{-1} - 2e+1) x^{-2e} 
\end{align*}
which proves \eqref{xexpansion} as required.
\end{proof}

\begin{remark}
The {\em Hermite polynomials} are defined by
\begin{equation}  \label{hermite}
H_N(x):=(-1)^Ne^{x^2}\frac{d^N}{dx^N}e^{-x^2}=\sum_{k=0}^{\lfloor\frac{N}{2}\rfloor}(-1)^k\binom{N}{2k}(2k-1)!!2^{N-k}x^{N-2k}
\end{equation}
So $H_N(x)$ is a degree $N$ polynomial in $x$, for example $H_0(x)=1$, $H_1(x)=2x$ and $H_2(x)=4x^2-2$.  They satisfy the Hermite equations
$$
\Big[\left(\frac{d}{dx}\right)^2-2x\frac{d}{dx}+2N\Big]H_N(x)=0.
$$

For any function $f:H_N\to\bc$ on $N\times N$ Hermitian matrices, define 
$$\langle f\rangle_N:=\frac{1}{Z_N}\int_{H_N}f(A)\exp\Big[-\frac{N}{2}\tr A^2\Big]dA,\quad Z_N=\int_{H_N}\exp\Big[-\frac{N}{2}\tr A^2\Big]dA$$
so in particular $\langle A_{ij}\rangle=0$ and $\langle |A_{ij}|^2\rangle=1/N$.

Then it is well-known that
\begin{equation}  \label{<det>}
\langle\det(x I-A)\rangle_N=(2N)^{-N/2}H_N(x\sqrt{\frac{N}{2}}).
\end{equation}
The following proof of \eqref{<det>} comes from \cite{FGaCou}.
\begin{align*}
\langle\det(x I-A)\rangle_N&=\langle\sum_{P\in S_N}\epsilon(P)\prod_{i=1}^N(x\delta_{i,P(i)}-A_{i,P(i)})\rangle_N=\langle\sum_{P\in S_N^{(2)}}\epsilon(P)\prod_{i=1}^N(x\delta_{i,P(i)}-A_{i,P(i)})\rangle_N\\
&=\sum_{k=0}^{\lfloor\frac{N}{2}\rfloor}(-1)^k\binom{N}{2k}(2k-1)!!x^{N-2k}N^{-k}=(2N)^{-N/2}H_N(x\sqrt{\frac{N}{2}})
\end{align*}
where the second equality uses the fact that the only non-zero contributions to the integral come from permutations with no cycles of length greater than 2, denoted by $S_N^{(2)}\subset S_N$.  For the third equality, each $P\in S_N^{(2)}$ is a product of say $N-2k$ fixed points and $k$ 2-cycles which each contribute a factor of $x$, respectively $1/N$. The factor $(-1)^k$ comes from the parity of $P$ and the factor $\binom{N}{2k}(2k-1)!!$ is the number of ways of choosing $N-2k$ fixed points and $k$ 2-cycles.  The final equality uses the formula \eqref{hermite} for the scaled Hermite polynomial.  One consequence of \eqref{<det>} is that 
$\langle\det(x I-A)\rangle_N$ satisfies the (scaled) Hermite equation since
$$\Big[\left(\frac{1}{N}\frac{d}{dx}\right)^2-x\frac{1}{N}\frac{d}{dx}+1\Big]H_N(x\sqrt{\frac{N}{2}})=0.$$
Moreover, for $\h=1/N$
$$\langle\det(x-A)\rangle_N=\psi(x,\h)$$
where the left hand side is considered as a function of $N$.  This is obtained by comparing the exact expression \eqref{xexpansion} for $\opsi(x,\h)$ with that for $\langle\det(I-Ax^{-1})\rangle_N$.  This identification also gives a proof of the well-known semi-classical limit:
$$\lim_{N\to\infty}\frac{1}{N}\frac{d}{dx}\log\langle\det(x I-A)\rangle=\frac{1}{2\pi}\int^2_{-2}\frac{\sqrt{4-t^2}}{x-t}dt,\quad x\not\in[-2,2].
$$
%which shows that the asymptotic behaviour of the zeros of the scaled Hermite polynomials are distributed like $\frac{1}{2\pi}\sqrt{4-t^2}$.::

\begin{remark}
Another proof of Theorem~\ref{thm:qcurve} appears in \cite{ACMPMod}.  It starts from the fact that \eqref{genbelyi} is related to correlation functions of a matrix model, known as the Kontsevich-Penner matrix model \cite{CMaMul} which is an integral over $N\times N$ Hermitian matrices with a potential that in some sense generalises the integral representation of Hermite polynomials.  The wave function defined in \eqref{wavefun} arises as the specialisation of the partition function of the Kontsevich-Penner matrix model for $1\times 1$ matrices!  The proof that it satisfies the differential equation \eqref{qde} is immediate because the integral is simply one dimensional.
\end{remark}

\end{remark}

\end{document}